\documentclass[11pt]{article}

\usepackage[letterpaper,margin=1in]{geometry}

\usepackage{amsmath,amssymb,amsthm,mathtools}
\usepackage{graphicx}
\usepackage{xcolor}
\usepackage{microtype}
\usepackage{comment}

\usepackage[colorlinks=true,linkcolor=blue,citecolor=blue,urlcolor=blue]{hyperref}
\usepackage[nameinlink,capitalise,noabbrev]{cleveref}

\usepackage{enumitem}
\usepackage{booktabs}
\usepackage{algorithm}
\usepackage{algorithmic}

\newtheorem{theorem}{Theorem}[section]
\newtheorem{lemma}[theorem]{Lemma}
\newtheorem{proposition}[theorem]{Proposition}
\newtheorem{corollary}[theorem]{Corollary}

\newtheorem{construction}[theorem]{Construction}

\theoremstyle{definition}
\newtheorem{definition}[theorem]{Definition}

\theoremstyle{remark}
\newtheorem{remark}[theorem]{Remark}

\newcommand{\F}{\mathbb{F}}

\newcommand{\BW}{\mathrm{BW}}
\newcommand{\IO}{\mathrm{IO}}
\newcommand{\PG}{\mathrm{PG}}

\newcommand{\avg}{\mathrm{avg}}
\newcommand{\rank}{\mathrm{rank}}
\newcommand{\nz}{\mathrm{nz}}

\newcommand{\col}{\mathrm{col}}
\newcommand{\Span}{\mathrm{span}}
\newcommand{\GL}{\mathrm{GL}}

\begin{document}

\begin{center}
  {\LARGE\bfseries
    Linear Exact Repair in MDS Array Codes:\\[0.3em]
    A General Lower Bound and Its Attainability\par}
\end{center}

\vspace{0.6em}
\begin{center}
  {\large
    Hai Liu\textsuperscript{2}
    \quad
    Huawei Wu\textsuperscript{1,*}
  }\\[0.55em]
  {\normalsize
    \textsuperscript{1}Shanghai Fintelli Box Technology Co., Ltd., Shanghai, 200127, China\\[0.35em]
    \textsuperscript{2}School of Software Engineering, East China Normal University, Shanghai, 200062, China\\[0.35em]
    \textsuperscript{*}Corresponding author.
    \href{mailto:wuhuawei1996@gmail.com}{\texttt{wuhuawei1996@gmail.com}}
  }
\end{center}

\begin{abstract}
 For an \((n,k,\ell)\) MDS array code over \(\F_q\), how small can the repair bandwidth and repair I/O be under linear exact repair? We study this question in the regime where the field size \(q\), the redundancy \(r=n-k\), and the sub-packetization level \(\ell\) are fixed, while the code length \(n\) varies, and we develop a geometric approach to this setting. Our starting point is an intrinsic reformulation of linear exact repair for MDS array codes in terms of subspace intersections and, for repair I/O, the projective point configurations induced by a parity-check realization.

This viewpoint yields a simple projective counting argument establishing the general lower bound
\[
\beta_{\avg},\beta_{\max},\gamma_{\avg},\gamma_{\max}
\;\ge\;
\ell(n-1)-\frac{q^{(r-1)\ell}-1}{q-1}
\]
for linear exact repair of every \((n,k,\ell)\) MDS array code over \(\F_q\) with redundancy \(r=n-k\ge 2\). To our knowledge, this is the first lower bound of this form that applies to arbitrary redundancy \(r\ge 2\) and sub-packetization level \(\ell\). At first glance, the projective counting bound appears rather coarse and therefore unlikely to be attained. We prove that this intuition is correct whenever \(r\ge 3\) and \(\ell\ge 2\).

For \(r=2\), the picture changes completely. Using Desarguesian spreads from finite geometry, we construct MDS array codes that attain the bound over a broad interval of code lengths, up to the maximum possible length \(q^\ell+1\), and do so simultaneously for both repair bandwidth and repair I/O. In the smallest nontrivial case \((r,\ell)=(2,2)\), we also prove a converse within the regular-spread model. 

Together, these results identify a uniform obstruction governing linear exact repair and show that, in the two-parity case, this obstruction is tight.
  \end{abstract}

%
\newpage
\tableofcontents
\newpage

\section{Introduction}

Erasure-coded storage systems must routinely recover from the temporary unavailability or failure of a single storage node. For an $(n,k)$ MDS-coded system, this leads to a basic algorithmic question: how efficiently can one repair a missing node while preserving the optimal redundancy--reliability trade-off of MDS coding? A central measure of repair efficiency is the \emph{repair bandwidth}, namely the total amount of information downloaded from helper nodes during repair. A key structural parameter governing repair is the \emph{sub-packetization} level $\ell$: larger sub-packetization enables finer-grained repair and can significantly reduce repair bandwidth~\cite{ramkumar2022codes}.

The regenerating-code framework identified repair bandwidth as a fundamental quantity and established the \emph{cut-set bound} for single-node repair. At the \emph{minimum-storage point} of this trade-off, one obtains \emph{minimum-storage regenerating (MSR) codes}, which preserve the MDS property and achieve the information-theoretically optimal repair bandwidth~\cite{dimakis2010network}. The difficulty is that this optimum is provably expensive: in the high-rate regime, exact-repair MSR codes require exponential, or at least very large, sub-packetization~\cite{balaji2018tight,alrabiah2019exponential,balaji2022lower}. In practice, such fine-grained partitioning can lead to fragmented, often non-contiguous, disk accesses~\cite{gan2025revisiting,shen2025survey}. This has motivated a broad line of work on \emph{MDS array codes} with small, or even constant, sub-packetization, which seek the best achievable repair bandwidth without insisting on the MSR point~\cite{rashmi2017piggybacking,kralevska2018hashtag,rawat2017msr,kralevska2018explicit,ramkumar2022codes,ramkumar2025varepsilon}. A central open direction is to understand the trade-off between repair bandwidth, sub-packetization, and field size for MDS array codes~\cite[Open Problem~9]{ramkumar2022codes}.

In this paper, we consider \emph{linear exact repair} for $(n,k,\ell)$ MDS array codes over the finite field $\F_q$ of size $q$, with redundancy $r=n-k$, in the regime where the field size $q$, the redundancy $r$, and the sub-packetization level $\ell$ are fixed while the code length $n$ varies. Besides repair bandwidth, we also study the \emph{repair I/O}, defined as the total amount of information accessed at helper nodes during repair. This quantity has received increasing attention in recent years~\cite{dau2018repair,li2019costs,liu2024formula,liu2025calculating}. Our goal is to understand the fundamental limitations of repair bandwidth and repair I/O in this regime, and in particular to prove general lower bounds on both quantities.

The closest prior work in this direction is due to Zhang, Li, and Hu~\cite{zhang2025optimal}, who analyzed the special case $(r,\ell)=(2,2)$ and obtained explicit lower bounds for repair bandwidth and repair I/O that are nearly sharp in certain field-size-dependent short-length regimes. Their results reveal a delicate small-parameter phenomenon, but do not explain the general obstruction governing repair complexity beyond $(r,\ell)=(2,2)$.

Our starting point is that the usual matrix description of linear exact repair is convenient for bookkeeping but tends to obscure the underlying geometry of the repair constraints. We show that, for repair bandwidth, linear exact repair admits an intrinsic reformulation in terms of intersections between the node subspaces and feasible repair subspaces. For repair I/O, the same viewpoint persists in a refined form: one must keep track not only of the node subspaces themselves, but also of the projective column points arising from a chosen parity-check realization.

This intrinsic subspace formulation naturally yields, via a simple counting argument in projective space, a general lower bound on repair bandwidth for linear exact repair in MDS array codes. To our knowledge, this is the first lower bound of this kind that applies to arbitrary redundancy $r\ge 2$ and sub-packetization level $\ell$. By the pointwise inequality $\IO\ge \BW$, the same quantity also gives a general lower bound on repair I/O.

At first glance, the bound appears rather coarse, and one might expect it to be rarely attained. We show that this intuition is correct once $r\ge 3$ and $\ell\ge 2$: in that regime, equality never occurs. Surprisingly, the two-parity case is fundamentally different. When $r=2$, constructions arising from \emph{Desarguesian spreads} in finite geometry attain equality over a broad interval of admissible code lengths, reaching all the way to the maximum possible length $q^\ell+1$; with suitable parity-check realizations, these constructions are simultaneously optimal for both repair bandwidth and repair I/O.

Taken together, these results provide not only general lower bounds on repair bandwidth and repair I/O for linear exact repair in MDS array codes, but also a new geometric viewpoint on these quantities. The projective counting bound provides a uniform constraint for arbitrary $r\ge 2$ and $\ell$, while in the two-parity case the same viewpoint yields a remarkably clean finite-geometric explanation of sharpness, with optimality witnessed by explicit extremal configurations. The two-parity case is also practically important, since it corresponds to a very low redundancy level and is therefore adopted in widely used double-erasure-tolerant storage systems (e.g., RAID-6 \cite{blaum1995evenodd}, Tencent Ultra-Cold Storage \cite{intel2017tencent}).

\subsection{Main Results}

We now state the main results of the paper. For an $(n,k,\ell)$ MDS array code $\mathcal{C}$ over $\F_q$, with redundancy $r:=n-k$, let
\[
  \beta_{\avg}(\mathcal{C}),\ \beta_{\max}(\mathcal{C}),\ \gamma_{\avg}(\mathcal{C}),\ \gamma_{\max}(\mathcal{C})
\]
denote, respectively, the average and worst-case repair bandwidth, and the average and worst-case repair I/O, under linear exact repair.

Our first theorem is a general lower bound, valid for arbitrary redundancy $r\ge 2$ and sub-packetization level $\ell$. For convenience, set
\[
  T_{r,\ell}(q):=\frac{q^{(r-1)\ell}-1}{q-1}.
\]

\begin{theorem}[Projective counting bound]\label{thm:projective-counting-bound}
  Let $\mathcal{C}$ be an $(n,k,\ell)$ MDS array code over $\F_q$ with redundancy $r=n-k\ge 2$. Then
  \[
    \beta_{\avg}(\mathcal{C}),\ \beta_{\max}(\mathcal{C}),\ \gamma_{\avg}(\mathcal{C}),\ \gamma_{\max}(\mathcal{C})
    \ \ge\ \ell(n-1)-T_{r,\ell}(q).
  \]
\end{theorem}

The next result shows that the projective counting bound is never attained once $r\ge 3$ and $\ell\ge 2$.

\begin{theorem}\label{thm:strictness}
  Assume that $r\ge 3$ and $\ell\ge 2$. Then for every $(n,k,\ell)$ MDS array code $\mathcal{C}$ over $\F_q$,
  \[
    \beta_{\avg}(\mathcal{C}),\ \beta_{\max}(\mathcal{C}),\ \gamma_{\avg}(\mathcal{C}),\ \gamma_{\max}(\mathcal{C})
    \ >\ \ell(n-1)-T_{r,\ell}(q).
  \]
\end{theorem}

In contrast, when $r=2$, the projective counting bound is attained over a broad interval of code lengths by constructions arising from Desarguesian spreads in finite geometry.

\begin{theorem}\label{thm:two-parity-attainability-general}
  Assume that $q\ge 3$ and $\ell\ge 2$, and let
  \[
    t_\ell(q):=T_{2,\ell}(q)=\frac{q^\ell-1}{q-1}.
  \]
  Then for every integer $n$ satisfying
  \[
    \min\{2t_\ell(q),\,3t_\ell(q)-6\}\le n\le q^\ell+1,
  \]
  there exists an $(n,n-2,\ell)$ MDS array code $\mathcal{C}$ over $\F_q$ such that
  \[
    \beta_{\avg}(\mathcal{C})=\beta_{\max}(\mathcal{C})
    =\gamma_{\avg}(\mathcal{C})=\gamma_{\max}(\mathcal{C})
    =\ell(n-1)-t_\ell(q).
  \]
  In particular, $\mathcal{C}$ attains the projective counting bound simultaneously for both repair bandwidth and repair I/O.
\end{theorem}

Finally, in the smallest nontrivial case $(r,\ell)=(2,2)$, we prove a converse to the preceding two-parity construction theorem within the \emph{regular-spread} model.

\begin{theorem}\label{thm:two-parity-attainability-regular-spread}
  Assume that $q\ge 3$ and $r=\ell=2$. Fix a regular spread $\mathcal{S}$ of $\PG(3,q)$. Then the following are equivalent:
  \begin{enumerate}[label=(\arabic*)]
    \item $\min\{2q+2,\,3q-3\}\le n\le q^2+1$.
    \item There exists an $(n,n-2,2)$ MDS array code $\mathcal{C}$ over $\F_q$ whose induced node lines all lie in $\mathcal{S}$, and such that at least one (and hence all) of
          \[
            \beta_{\avg}(\mathcal{C}),\ \beta_{\max}(\mathcal{C}),\ \gamma_{\avg}(\mathcal{C}),\ \gamma_{\max}(\mathcal{C})
          \]
          attains the projective counting bound.
  \end{enumerate}
\end{theorem}

Theorems~\ref{thm:projective-counting-bound} and~\ref{thm:strictness} are proved in Section~\ref{sec:projective-counting-bound}, while Theorems~\ref{thm:two-parity-attainability-general} and~\ref{thm:two-parity-attainability-regular-spread} are proved in Section~\ref{sec:two-parity}.

\subsection{Prior and Related Work}

We briefly discuss prior work most closely related to the present paper. A natural starting point is~\cite[Open Problem~9]{ramkumar2022codes}, which formulates the general trade-off question between repair bandwidth, sub-packetization, and field size for vector MDS codes.

A related line of work studies fine-grained repair for scalar MDS codes, especially Reed--Solomon codes over an extension field repaired over a subfield. Guruswami and Wootters~\cite{guruswami2016repairing} initiated the study of repair bandwidth in this setting, and Dau and Milenkovic~\cite{dau2017optimal} sharpened the resulting bandwidth bounds for Reed--Solomon codes. More recently, repair I/O in the same scalar-MDS framework has also been studied:~\cite{dau2018repair} gave the first nontrivial repair-I/O lower bound for full-length Reed--Solomon codes with two parity nodes, and subsequent works~\cite{li2019costs,liu2024formula,liu2025calculating} extended and refined these bounds. These works are closely related in spirit, since they also seek lower bounds on fine-grained repair complexity, but they concern a different model: the code remains scalar over a large field, whereas the present paper studies genuine MDS array codes with fixed redundancy and fixed sub-packetization.

The closest prior work to ours is the recent paper of Zhang, Li, and Hu~\cite{zhang2025optimal}, which studies the special case $(r,\ell)=(2,2)$. By a delicate combinatorial analysis, they derive explicit lower bounds for repair bandwidth and repair I/O that depend only on the code length $n$. They further show that these bounds are nearly sharp in certain short-length regimes by constructing matching codes for lengths below $q$. This yields a precise understanding of the smallest parameter regime, but the analysis is inherently tied to $(r,\ell)=(2,2)$ and does not identify the obstruction governing repair complexity for arbitrary $r$ and $\ell$. Moreover, as the code length grows, their bounds depending only on $n$ no longer capture the governing constraint, which in our setting turns out to depend essentially on the size of the underlying field.

Another recent direction studies the same $(r,\ell)=(2,2)$ regime under the additional degraded-read-friendly (DRF) restriction. An earlier work~\cite{wu2021achievable} derives a lower bound on the optimal access bandwidth for DRF MDS array codes in this setting, and also gives a matching construction. More recently, Li and Tang~\cite{li2025efficient} develop explicit DRF constructions, highlighting in particular two families with $n=4m$ and $n=3m$, and emphasize improved repair bandwidth or rebuilding access relative to previously known constructions. In particular, one of their constructions is asymptotically optimal with respect to the average-case repair-bandwidth lower bound of Zhang, Li, and Hu~\cite{zhang2025optimal}.

\section{Linear Exact Repair and Its Intrinsic Subspace Formulation}\label{sec:intrinsic}

\subsection{MDS Array Codes and Linear Exact Repair}\label{subsection:mds}

We begin with the standard matrix description of an MDS array code and of linear exact repair.
Throughout the paper, for a positive integer $m$, we write $[m]:=\{1,2,\dots,m\}$.

Fix integers $n,k,\ell\in\mathbb{N}_+$ with $1\le k<n$, and write $r:=n-k$. Here $n$ is the code length, $k$ is the dimension parameter, $\ell$ is the sub-packetization level,
and $r$ is the redundancy.

An \emph{$(n,k,\ell)$ linear array code} over $\F_q$ is an $\F_q$-linear subspace
\[
  \mathcal{C}\le (\F_q^\ell)^n
\]
of dimension $k\ell$ over $\F_q$.
Its elements are written as
\[
  C=(C_1,\dots,C_n), \qquad C_i\in \F_q^\ell.
\]
We refer to the vectors $C_1,\dots,C_n$ as the \emph{blocks} of the codeword $C$.
When interpreted in the distributed-storage setting, block $C_i$ is stored in node $i$.
Thus each node stores $\ell$ subsymbols over $\F_q$.

A convenient description of $\mathcal{C}$ is via a block parity-check matrix
\[
  H=[\,H_1\ H_2\ \cdots\ H_n\,]\in \F_q^{r\ell\times n\ell},
  \qquad H_i\in \F_q^{r\ell\times \ell},
\]
of full row rank $r\ell$, so that
\[
  \mathcal{C}
  =
  \ker(H)
  =
  \Bigl\{
    (C_1,\dots,C_n)\in (\F_q^\ell)^n:
    \sum_{i=1}^n H_iC_i=0
  \Bigr\}.
\]
We say that $\mathcal{C}$ is \emph{MDS} if for every subset $I\subseteq [n]$ with $|I|=r$, the
square block matrix
\[
  H_I:=[\,H_i\,]_{i\in I}\in \F_q^{r\ell\times r\ell}
\]
is invertible. In other words, any $k=n-r$ blocks determine the entire codeword. In the distributed-storage setting, this means that any $r$ erased nodes can be recovered from the
remaining $k$ nodes.

From this point on, we always assume that $\mathcal{C}$ is MDS. In particular, each block
$H_i$ has full column rank $\ell$: indeed, for any $i\in[n]$, the block $H_i$ appears inside
some invertible matrix $H_I$ with $|I|=r$, and hence its $\ell$ columns must be linearly independent.

We then recall the standard notion of \emph{linear exact repair}. Suppose that node $i\in[n]$ fails, so that the block $C_i$ is missing. A linear repair scheme for node $i$ is specified by a matrix $M\in \F_q^{\ell\times r\ell}$, which produces $\ell$ new linear combinations of the parity-check equations:
\[
  \sum_{j=1}^n (MH_j)C_j=0.
\]
If the matrix $MH_i\in\F_q^{\ell\times \ell}$ is invertible, then these $\ell$ equations determine
the unknown block $C_i$ uniquely, since we may rearrange them as
\[
  C_i
  =
  -(MH_i)^{-1}\sum_{j\ne i}(MH_j)C_j.
\]
Thus node $i$ can be recovered provided that, for each helper node $j\ne i$, we know the vector
\[
  (MH_j)C_j.
\]
In this sense, helper node $j$ contributes the linear image $(MH_j)C_j$ of its stored block,
and the failed block $C_i$ is reconstructed by combining these contributions from all helper nodes.
Accordingly, whenever $MH_i$ is invertible, we say that $M$ \emph{repairs node $i$}.

This gives rise to two basic repair-cost measures.

\smallskip
\noindent\textbf{Repair bandwidth.} For a fixed failed node $i$ and a repair matrix $M$ with $MH_i$ invertible, define
\[
  \BW_i(M):=\sum_{j\ne i}\rank(MH_j).
\]
This is the total number of $\F_q$-symbols downloaded from the helper nodes during the repair of
node $i$.

\smallskip
\noindent\textbf{Repair I/O.} For a matrix $A$, let $\nz(A)$ denote the number of nonzero columns of $A$.
Since $(MH_j)C_j$ depends only on those coordinates of $C_j$ corresponding to nonzero columns of
$MH_j$, we define
\[
  \IO_i(M):=\sum_{j\ne i}\nz(MH_j).
\]
Thus $\IO_i(M)$ measures the total number of subsymbols accessed at the helper nodes during repair. Clearly, one has $\IO_i(M)\ge \BW_i(M)$, since $\nz(A)\ge \rank(A)$ for every matrix $A$.

For each node $i\in[n]$, let
\[
  \mathcal{M}_i
  :=
  \{\,M\in \F_q^{\ell\times r\ell}: MH_i \text{ is invertible}\,\}
\]
be the set of linear repair matrices that can repair node $i$.
We then define the optimal per-node repair bandwidth and repair I/O by
\[
  \beta_i(\mathcal{C})
  :=
  \min_{M\in \mathcal{M}_i}\BW_i(M),
  \qquad
  \gamma_i(\mathcal{C})
  :=
  \min_{M\in \mathcal{M}_i}\IO_i(M).
\]

Aggregating over all failed nodes, we obtain the average and worst-case parameters
\[
  \beta_{\avg}(\mathcal{C})
  :=
  \frac{1}{n}\sum_{i=1}^n \beta_i(\mathcal{C}),
  \qquad
  \beta_{\max}(\mathcal{C})
  :=
  \max_{i\in[n]}\beta_i(\mathcal{C}),
\]
and
\[
  \gamma_{\avg}(\mathcal{C})
  :=
  \frac{1}{n}\sum_{i=1}^n \gamma_i(\mathcal{C}),
  \qquad
  \gamma_{\max}(\mathcal{C})
  :=
  \max_{i\in[n]}\gamma_i(\mathcal{C}).
\]

The matrix formulation above is standard and convenient for bookkeeping. However, for our purposes it hides the geometric content of the repair constraints. In the next subsection we recast this repair problem in an intrinsic subspace language, which is the framework used throughout the rest of the paper.

\subsection{An Intrinsic Subspace Formulation}

The matrix model from Subsection~\ref{subsection:mds} admits a reversible reformulation in terms of subspaces of
$\mathbb{V}:=\F_q^{r\ell}$, together with, for repair I/O, distinguished projective point sets inside those subspaces. This is the framework used throughout the rest of the paper.

Write each parity-check block as $H_i=[\,h_{i,1}\ \cdots\ h_{i,\ell}\,]$ with $h_{i,t}\in\mathbb{V}$,
and define
\[
  \mathcal{H}_i:=\col(H_i)=\Span_{\F_q}\{h_{i,1},\dots,h_{i,\ell}\}\le \mathbb{V}.
\]
We will informally refer to \(\mathcal{H}_i\) as the \emph{node subspace} associated with node \(i\).

Since $\mathcal{C}$ is MDS, each block $H_i$ has rank $\ell$, so $\dim(\mathcal{H}_i)=\ell$.
Moreover, the MDS condition is equivalent to requiring that, for every $J\subseteq[n]$ with $|J|=r$,
one has $\sum_{j\in J}\mathcal{H}_j=\mathbb{V}$; since each $\mathcal{H}_j$ has dimension $\ell$ and
$\dim(\mathbb{V})=r\ell$, this is equivalent to the sum being direct. Conversely, given
$\ell$-dimensional subspaces $\mathcal{H}_1,\dots,\mathcal{H}_n\le\mathbb{V}$ satisfying this
$r$-wise spanning condition, one may choose full-column-rank matrices $H_i$ with
$\col(H_i)=\mathcal{H}_i$; then the block matrix $H=[\,H_1\ \cdots\ H_n\,]$ defines an
$(n,k,\ell)$ MDS array code over $\F_q$.

Now fix a failed node $i\in[n]$, and let $M\in\F_q^{\ell\times r\ell}$ be a repair matrix for node $i$.
The intrinsic object associated with $M$ is its kernel $W:=\ker(M)\le\mathbb{V}$. Since $MH_i$ is
invertible, the matrix $M$ has rank $\ell$, and hence $\dim(W)=r\ell-\ell$. Moreover, $MH_i$ is
invertible if and only if the restriction $M|_{\mathcal{H}_i}:\mathcal{H}_i\to\F_q^\ell$ is injective;
since $\dim(\mathcal{H}_i)=\ell$, this is equivalent to $\ker(M)\cap\mathcal{H}_i=\{0\}$, that is, to
$W\cap\mathcal{H}_i=\{0\}$. Conversely, every subspace $W\le\mathbb{V}$ with $\dim(W)=r\ell-\ell$ and
$W\cap\mathcal{H}_i=\{0\}$ is the kernel of some matrix $M\in \F_q^{\ell\times r\ell}$ of rank $\ell$, and any such $M$ repairs node $i$. This motivates the feasible family
\[
  \mathcal{W}_i:=\{\,W\le\mathbb{V}:\dim(W)=r\ell-\ell,\ W\cap\mathcal{H}_i=\{0\}\,\}.
\]
We will informally refer to elements of \(\mathcal{W}_i\) as \emph{repair subspaces} for node \(i\).

For such a repair matrix $M$, with $W=\ker(M)$, and for each $j\ne i$, one has
\[
  \rank(MH_j)=\ell-\dim(W\cap\mathcal{H}_j),
\]
and hence
\[
  \BW_i(M)=\sum_{j\ne i}\rank(MH_j)
  =\ell(n-1)-\sum_{j\ne i}\dim(W\cap\mathcal{H}_j).
\]
Thus minimizing repair bandwidth is equivalent to maximizing the total intersection dimension with the
helper node subspaces. We therefore define
\[
  \alpha_i:=\max_{W\in\mathcal{W}_i}\sum_{j\ne i}\dim(W\cap\mathcal{H}_j).
\]
With $\alpha_{\avg}:=\frac1n\sum_{i=1}^n\alpha_i$ and $\alpha_{\min}:=\min_{i\in[n]}\alpha_i$, we obtain
\[
  \beta_i(\mathcal{C})=\ell(n-1)-\alpha_i,\qquad
  \beta_{\avg}(\mathcal{C})=\ell(n-1)-\alpha_{\avg},\qquad
  \beta_{\max}(\mathcal{C})=\ell(n-1)-\alpha_{\min}.
\]

Unlike repair bandwidth, repair I/O is not determined by the node subspaces
$\mathcal{H}_1,\dots,\mathcal{H}_n$ alone: it also depends on the individual columns inside each block
$H_i$. We therefore record the set of projective column points
\[
  X_i:=\{\langle h_{i,1}\rangle,\dots,\langle h_{i,\ell}\rangle\}\subseteq\mathbb{P}(\mathcal{H}_i),
\]
where, for a nonzero subspace $U\le\mathbb{V}$, $\mathbb{P}(U)$ denotes its projectivization, namely
the set of $1$-dimensional subspaces of $U$, and $\langle h\rangle$ denotes the $1$-dimensional
subspace spanned by a nonzero vector $h$. Since the columns of $H_i$ are linearly independent, the
points in $X_i$ are distinct and span $\mathcal{H}_i$. Conversely, any set
$X_i\subseteq\mathbb{P}(\mathcal{H}_i)$ of $\ell$ distinct points spanning $\mathcal{H}_i$ can be
realized by choosing one nonzero representative from each point of $X_i$ as a column of $H_i$.

For $W\in\mathcal{W}_i$ and $j\ne i$, let
\[
  z_j(W):=\bigl|\{\,t\in[\ell]:h_{j,t}\in W\,\}\bigr|.
\]
Since $h\in W$ if and only if $\langle h\rangle\in\mathbb{P}(W)$, this may also be written as
\[
  z_j(W)=|X_j\cap\mathbb{P}(W)|.
\]
Since a column of $MH_j$ is zero exactly when the corresponding column of $H_j$ lies in $W$, we have
\[
  \IO_i(M)=\ell(n-1)-\sum_{j\ne i}z_j(W).
\]
Thus minimizing repair I/O is equivalent to maximizing the total number of helper-column points
captured by $W$. We therefore define
\[
  \lambda_i:=\max_{W\in\mathcal{W}_i}\sum_{j\ne i}z_j(W).
\]
With $\lambda_{\avg}:=\frac1n\sum_{i=1}^n\lambda_i$ and
$\lambda_{\min}:=\min_{i\in[n]}\lambda_i$, we obtain
\[
  \gamma_i(\mathcal{C})=\ell(n-1)-\lambda_i,\qquad
  \gamma_{\avg}(\mathcal{C})=\ell(n-1)-\lambda_{\avg},\qquad
  \gamma_{\max}(\mathcal{C})=\ell(n-1)-\lambda_{\min}.
\]

Thus repair bandwidth is governed by the intersection dimensions $W\cap\mathcal{H}_j$, whereas repair I/O depends on the captured column points $X_j\cap\mathbb{P}(W)$. The next section uses
this intrinsic formulation to derive the projective counting lower bound.

\section{The Projective Counting Bound}
\label{sec:projective-counting-bound}

\subsection{Deriving the Projective Counting Bound}

We continue to use the notation of Section~\ref{sec:intrinsic}, and recall that
\[
  T_{r,\ell}(q):=\frac{q^{(r-1)\ell}-1}{q-1}.
\]
We first bound the quantities $\alpha_i$ by a simple counting argument in projective space. The lower bound
for repair I/O then follows immediately from the inequality
$\gamma_i(\mathcal{C})\ge \beta_i(\mathcal{C})$.

\begin{lemma}\label{lem:projective-counting-bound}
  Assume that $r\ge 2$. Fix $i\in[n]$ and let $W\in\mathcal{W}_i$. Then
  \[
    \sum_{j\ne i}\dim(W\cap \mathcal{H}_j)\le T_{r,\ell}(q).
  \]
\end{lemma}

\begin{proof}
  Set $t_j:=\dim(W\cap \mathcal{H}_j)$ for $j\ne i$. Since $\mathcal{C}$ is MDS, for every
  $J\subseteq[n]$ with $|J|=r$, the sum $\sum_{j\in J}\mathcal{H}_j$ is direct; in particular,
  $\mathcal{H}_a\cap \mathcal{H}_b=\{0\}$ for all distinct $a,b\in[n]$. Hence
  $(W\cap \mathcal{H}_a)\cap (W\cap \mathcal{H}_b)=\{0\}$ for all distinct $a,b\ne i$, so the
  projective point sets $\mathbb{P}(W\cap \mathcal{H}_a)$ and $\mathbb{P}(W\cap \mathcal{H}_b)$ are
  pairwise disjoint, where we adopt the convention that $\mathbb{P}(0)=\emptyset$.

  Since $\dim(W)=(r-1)\ell$, we have
  \[
    \sum_{j\ne i} |\mathbb{P}(W\cap \mathcal{H}_j)| \le |\mathbb{P}(W)| = T_{r,\ell}(q).
  \]
  On the other hand, for every nonnegative integer $t$,
  \[
    t\le \frac{q^t-1}{q-1}=|\mathbb{P}(\F_q^t)|.
  \]
  Therefore $t_j\le |\mathbb{P}(W\cap \mathcal{H}_j)|$ for each $j\ne i$, and summing gives
  \[
    \sum_{j\ne i}\dim(W\cap \mathcal{H}_j)=\sum_{j\ne i} t_j \le T_{r,\ell}(q).
  \]
\end{proof}

\begin{remark}\label{rem:necessary-length-for-equality}
  In the setting of Lemma~\ref{lem:projective-counting-bound}, equality can hold only if
  \[
    n\ge 1+T_{r,\ell}(q).
  \]
  Indeed, if equality holds and $t_j:=\dim(W\cap \mathcal{H}_j)$ for $j\ne i$, then each $t_j$ must lie
  in $\{0,1\}$. For if some $t_j\ge 2$, then
  \[
    t_j<\frac{q^{t_j}-1}{q-1}=|\mathbb{P}(W\cap \mathcal{H}_j)|,
  \]
  so the inequality
  \[
    \sum_{j\ne i} t_j\le \sum_{j\ne i} |\mathbb{P}(W\cap \mathcal{H}_j)|
  \]
  would be strict, contradicting equality in Lemma~\ref{lem:projective-counting-bound}. Hence
  \[
    T_{r,\ell}(q)=\sum_{j\ne i} t_j
    =\#\{\,j\ne i:\dim(W\cap \mathcal{H}_j)=1\,\}\le n-1.
  \]
\end{remark}

The projective counting bound for repair bandwidth and repair I/O now follows immediately.

\begin{proof}[Proof of Theorem~\ref{thm:projective-counting-bound}]
  By Lemma~\ref{lem:projective-counting-bound}, one has $\alpha_i\le T_{r,\ell}(q)$ for every $i\in[n]$.
  Hence
  \[
    \beta_i(\mathcal{C})=\ell(n-1)-\alpha_i \ge \ell(n-1)-T_{r,\ell}(q)
  \]
  for every $i\in[n]$. Averaging over $i$ and taking the maximum over $i$ yield
  \[
    \beta_{\avg}(\mathcal{C}),\ \beta_{\max}(\mathcal{C}) \ge \ell(n-1)-T_{r,\ell}(q).
  \]
  Since $\gamma_i(\mathcal{C})\ge \beta_i(\mathcal{C})$ for every $i\in[n]$, the same lower bound
  also holds for $\gamma_{\avg}(\mathcal{C})$ and $\gamma_{\max}(\mathcal{C})$.
\end{proof}

\subsection{Strictness for \texorpdfstring{\(r\ge 3\)}{r >= 3}}

The projective counting bound arises from a rather coarse counting argument, so one might suspect that equality should be exceptional. We now show that this is indeed the case once $r\ge 3$ and $\ell\ge 2$: in that regime, the bound is never attained. 

We begin with a general upper bound on the size of a family of $\ell$-subspaces satisfying the
$r$-wise spanning condition, or equivalently, on the length of an MDS array code.

\begin{lemma}\label{lem:length-bound}
  Assume that $r\ge 2$. Let $\mathcal{H}_1,\dots,\mathcal{H}_n\le \mathbb{V}$ be $\ell$-dimensional subspaces such that
  \[
    \sum_{j\in J}\mathcal{H}_j=\mathbb{V}
  \]
  for every $J\subseteq[n]$ with $|J|=r$. Then
  \[
    n\le q^\ell+r-1.
  \]
\end{lemma}

\begin{proof}
  Choose any subset $T\subseteq[n]$ with $|T|=r-2$, and set
  \[
    U:=\bigoplus_{j\in T}\mathcal{H}_j\le \mathbb{V}.
  \]
  Then $\dim(U)=(r-2)\ell$. Let $\overline{\mathbb{V}}:=\mathbb{V}/U$, so
  $\dim(\overline{\mathbb{V}})=2\ell$. For each $j\in[n]\setminus T$, define
  \[
    \overline{\mathcal{H}}_j:=(\mathcal{H}_j+U)/U\le \overline{\mathbb{V}}.
  \]
  Since the spanning condition implies that $U\cap \mathcal{H}_j=\{0\}$, each
  $\overline{\mathcal{H}}_j$ has dimension $\ell$.

  Now let $j_1\neq j_2$ lie in $[n]\setminus T$. Applying the spanning condition to
  $T\cup\{j_1,j_2\}$ gives
  \[
    \mathbb{V}=U\oplus \mathcal{H}_{j_1}\oplus \mathcal{H}_{j_2},
  \]
  and hence
  \[
    \overline{\mathbb{V}}=\overline{\mathcal{H}}_{j_1}\oplus \overline{\mathcal{H}}_{j_2}.
  \]
  In particular, the nonzero sets $\overline{\mathcal{H}}_j\setminus\{0\}$, for $j\in[n]\setminus T$,
  are pairwise disjoint subsets of $\overline{\mathbb{V}}\setminus\{0\}$. Counting nonzero vectors gives
  \[
    \bigl(n-(r-2)\bigr)(q^\ell-1)\le q^{2\ell}-1=(q^\ell-1)(q^\ell+1),
  \]
  and therefore $n-(r-2)\le q^\ell+1$, i.e.
  \[
    n\le q^\ell+r-1.
  \]
\end{proof}

We can now complete the proof that the projective counting bound is strict for $r\ge 3$.

\begin{proof}[Proof of Theorem~\ref{thm:strictness}]
  Suppose, for contradiction, that equality holds in at least one of the repair-bandwidth bounds, i.e.
  \[
    \beta_{\avg}(\mathcal{C})=\ell(n-1)-T_{r,\ell}(q)
    \qquad\text{or}\qquad
    \beta_{\max}(\mathcal{C})=\ell(n-1)-T_{r,\ell}(q).
  \]
  Since
  \[
    \beta_{\avg}(\mathcal{C})=\ell(n-1)-\alpha_{\avg},
    \qquad
    \beta_{\max}(\mathcal{C})=\ell(n-1)-\alpha_{\min},
  \]
  and $\alpha_i\le T_{r,\ell}(q)$ for every $i\in[n]$ by Lemma~\ref{lem:projective-counting-bound},
  it follows that $\alpha_i=T_{r,\ell}(q)$ for some $i\in[n]$. By the definition of $\alpha_i$, there
  therefore exists $W\in\mathcal{W}_i$ such that
  \[
    \sum_{j\ne i}\dim(W\cap \mathcal{H}_j)=T_{r,\ell}(q).
  \]
  Remark~\ref{rem:necessary-length-for-equality} now gives
  \[
    n\ge 1+T_{r,\ell}(q).
  \]
  On the other hand, Lemma~\ref{lem:length-bound} gives
  \[
    n\le q^\ell+r-1.
  \]
  These two inequalities are incompatible. Indeed, since $r\ge 3$ and $\ell\ge 2$, the sum
  \[
    T_{r,\ell}(q)=1+q+\cdots+q^{(r-1)\ell-1}
  \]
  strictly contains the terms
  \[
    1,\ q^\ell,\ q^{2\ell},\ \dots,\ q^{(r-2)\ell},
  \]
  and therefore
  \[
    T_{r,\ell}(q)>1+q^\ell+q^{2\ell}+\cdots+q^{(r-2)\ell}\ge q^\ell+r-2.
  \]
  Hence
  \[
    1+T_{r,\ell}(q)>q^\ell+r-1.
  \]
  This contradiction shows that
  \[
    \beta_{\avg}(\mathcal{C}),\ \beta_{\max}(\mathcal{C})
    \ >\ \ell(n-1)-T_{r,\ell}(q).
  \]
  Finally, since $\gamma_i(\mathcal{C})\ge \beta_i(\mathcal{C})$ for every $i\in[n]$, it follows that
  \[
    \gamma_{\avg}(\mathcal{C}),\ \gamma_{\max}(\mathcal{C})
    \ >\ \ell(n-1)-T_{r,\ell}(q)
  \]
  as well.
\end{proof}

\begin{remark}
  The proof of Lemma~\ref{lem:projective-counting-bound} is reminiscent of a sphere-packing argument: one packs pairwise disjoint projective point sets inside the ambient space $\mathbb{P}(W)$. In this light, the strictness result for $r\ge 3$ is analogous in spirit to the rigidity of perfect codes in classical coding theory (see \cite[Theorem~7.5.1]{lint1999introduction} for the binary case, and \cite{tietavainen1973nonexistence} for the general finite-field setting).
\end{remark}

\section{The Two-Parity Case}
\label{sec:two-parity}

\subsection{Constructions Attaining the Projective Counting Bound}
We now specialize to the two-parity case $r=2$. Set
\[
  t_\ell(q):=T_{2,\ell}(q)=\frac{q^\ell-1}{q-1}.
\]
Throughout this section, we identify $\mathbb{V}=\F_q^{2\ell}$ with $\F_{q^\ell}^2$ as $\F_q$-vector
spaces.

We begin by specifying a convenient family of candidate node subspaces. For each $c\in\F_{q^\ell}$, define
\[
  \mathcal{H}_c:=\{(x,cx):x\in\F_{q^\ell}\}\le\mathbb{V},
  \qquad
  \mathcal{H}_\infty:=\{(0,y):y\in\F_{q^\ell}\}\le\mathbb{V}.
\]
Let
\[
  \mathcal{S}_{\mathrm{Des}}
  :=\{\mathcal{H}_c:c\in\F_{q^\ell}\}\cup\{\mathcal{H}_\infty\}.
\]
This is the standard Desarguesian spread. Its definition and basic properties are recalled in Appendix~\ref{app:finite-geometry}. In particular, each member of \(\mathcal{S}_{\mathrm{Des}}\) is an \(\ell\)-dimensional
\(\F_q\)-subspace of \(\mathbb{V}\), and any two distinct members of
\(\mathcal{S}_{\mathrm{Des}}\) are complementary in \(\mathbb{V}\). Hence
they satisfy the \(r\)-wise spanning condition for \(r=2\).

We then introduce a family of candidate repair
subspaces whose intersections with the standard Desarguesian spread elements can be described explicitly. For each
$b\in\F_{q^\ell}^\times$, define
\[
  W_b:=\{(x,bx^q):x\in\F_{q^\ell}\}\le\mathbb{V}.
\]

\begin{lemma}\label{lem:two-parity-hit-set}
  Let
  \[
    \Sigma:=\ker\!\bigl(N_{\F_{q^\ell}/\F_q}\bigr)\subseteq \F_{q^\ell}^\times,
  \]
  where \(N_{\F_{q^\ell}/\F_q}:\F_{q^\ell}^\times\to\F_q^\times\) denotes the norm map. Then
  \(|\Sigma|=t_\ell(q)\). For every \(b\in\F_{q^\ell}^\times\) and every \(c\in\F_{q^\ell}\), one has
  \[
    W_b\cap \mathcal{H}_c\neq\{0\}
    \quad\Longleftrightarrow\quad
    c\in b\Sigma,
  \]
  and, whenever this holds,
  \[
    \dim(W_b\cap \mathcal{H}_c)=1.
  \]
  Moreover,
  \[
    W_b\cap \mathcal{H}_\infty=\{0\}.
  \]
\end{lemma}
\begin{proof}
  Since the norm map is a surjective group homomorphism, \(|\Sigma|=(q^\ell-1)/(q-1)=t_\ell(q)\). Also, \(\Sigma=\{u^{q-1}:u\in\F_{q^\ell}^\times\}\). Indeed, the map \(u\mapsto u^{q-1}\) has kernel
  \(\F_q^\times\), hence image of size \((q^\ell-1)/(q-1)\). This image is contained in
  \(\ker(N_{\F_{q^\ell}/\F_q})=\Sigma\), and the two sets have the same cardinality.

  Now \(W_b\cap\mathcal{H}_c\neq\{0\}\) if and only if there exists \(x\neq 0\) such that
  \((x,bx^q)=(x,cx)\), equivalently \(c=bx^{q-1}\). By the description of \(\Sigma\), this is
  equivalent to \(c\in b\Sigma\).

  If \(c=bu^{q-1}\) for some \(u\in\F_{q^\ell}^\times\), then
  \(W_b\cap\mathcal{H}_c=\{(\lambda u,\lambda cu):\lambda\in\F_q\}\), so
  \(\dim(W_b\cap\mathcal{H}_c)=1\).

  Finally, if \((0,y)\in W_b\), then \((0,y)=(x,bx^q)\) forces \(x=0\), hence \(y=0\). Thus
  \(W_b\cap\mathcal{H}_\infty=\{0\}\).
\end{proof}

We are now in a position to prove Theorem~\ref{thm:two-parity-attainability-general}.

\begin{proof}[Proof of Theorem~\ref{thm:two-parity-attainability-general}]
  We first treat the main range
\[
  2t_\ell(q)\le n\le q^\ell+1.
\]
This already proves Theorem~\ref{thm:two-parity-attainability-general} whenever $q\ge 5$ or
$\ell\ge 3$, since then $t_\ell(q)\ge 6$ and hence
\[
  \min\{2t_\ell(q),\,3t_\ell(q)-6\}=2t_\ell(q).
\] 
The only remaining cases are \((q,\ell,n)=(3,2,6),(3,2,7),(4,2,9)\), which will be handled in Appendix~\ref{app:exceptional-cases}.

  Let \(\Sigma:=\ker\!\bigl(N_{\F_{q^\ell}/\F_q}\bigr)\subseteq\F_{q^\ell}^\times\). Since the norm map is
  surjective and \(q\ge 3\), we may choose \(b_1,b_2\in\F_{q^\ell}^\times\) with distinct norms in
  \(\F_q^\times\). By Lemma~\ref{lem:two-parity-hit-set}, \(|\Sigma|=t_\ell(q)\), so the cosets
  \[
    C_1:=b_1\Sigma,\qquad C_2:=b_2\Sigma
  \]
  are disjoint subsets of \(\F_{q^\ell}^\times\), each of size \(t_\ell(q)\).

  Choose any \(\Omega\subseteq\F_{q^\ell}\cup\{\infty\}\) with \(|\Omega|=n\) and
  \(C_1\cup C_2\subseteq\Omega\); this is possible because \(2t_\ell(q)\le n\le q^\ell+1\). Enumerate the elements of \(\Omega\) as
  \[
    \Omega=\{c_1,\dots,c_n\},
  \]
  and set
  \[
    \mathcal{H}_i:=\mathcal{H}_{c_i}\qquad (i\in[n]).
  \]

  For each \(i\in[n]\), define
  \[
    W_i:=
    \begin{cases}
      W_{b_2}, & \text{if }\mathcal{H}_i=\mathcal{H}_c\text{ for some }c\in C_1,\\[2mm]
      W_{b_1}, & \text{otherwise.}
    \end{cases}
  \]
  By Lemma~\ref{lem:two-parity-hit-set}, \(W_{b_1}\) meets precisely the spread elements indexed by
  \(C_1\), and \(W_{b_2}\) precisely those indexed by \(C_2\), with every nonzero intersection having
  dimension \(1\). Since \(C_1\cap C_2=\emptyset\), the chosen \(W_i\) satisfies
  \(W_i\cap\mathcal{H}_i=\{0\}\), so \(W_i\in\mathcal{W}_i\). Moreover,
  \[
    \sum_{j\ne i}\dim(W_i\cap\mathcal{H}_j)=t_\ell(q)
  \]
  for every \(i\in[n]\), because exactly the \(t_\ell(q)\) subspaces indexed by the opposite coset
  contribute dimension \(1\), and all other intersections are trivial. Hence \(\alpha_i\ge t_\ell(q)\)
  for every \(i\). Since Lemma~\ref{lem:projective-counting-bound} gives the upper bound
  \(\alpha_i\le t_\ell(q)\), we obtain
  \[
    \alpha_i=t_\ell(q)\qquad\text{for all }i\in[n].
  \]

  We now choose the projective column sets \(X_i\subseteq\mathbb{P}(\mathcal{H}_i)\). If
  \(\mathcal{H}_i=\mathcal{H}_c\) with \(c\in C_1\), then \(\dim(W_{b_1}\cap\mathcal{H}_i)=1\), so
  \(\mathbb{P}(W_{b_1}\cap\mathcal{H}_i)\) is a single point of \(\mathbb{P}(\mathcal{H}_i)\); choose
  \(X_i\) to be any set of \(\ell\) distinct points spanning \(\mathcal{H}_i\) and containing this
  point. Similarly, if \(\mathcal{H}_i=\mathcal{H}_c\) with \(c\in C_2\), choose \(X_i\) to be any set
  of \(\ell\) distinct points spanning \(\mathcal{H}_i\) and containing the unique point of
  \(\mathbb{P}(W_{b_2}\cap\mathcal{H}_i)\). For all remaining node subspaces (including
  \(\mathcal{H}_\infty\), if selected), choose any set \(X_i\subseteq\mathbb{P}(\mathcal{H}_i)\) of
  \(\ell\) distinct points spanning \(\mathcal{H}_i\).

  By the converse statements in Section~\ref{sec:intrinsic}, the data
  \[
    \bigl(\mathcal{H}_i,X_i\bigr)_{i=1}^n
  \]
  are realized by an \((n,n-2,\ell)\) MDS array code \(\mathcal{C}\) over \(\F_q\).

  Since \(\alpha_i=t_\ell(q)\) for every \(i\), Section~\ref{sec:intrinsic} gives
  \[
    \beta_i(\mathcal{C})=\ell(n-1)-t_\ell(q)\qquad\text{for all }i\in[n].
  \]
  Hence
  \[
    \beta_{\avg}(\mathcal{C})=\beta_{\max}(\mathcal{C})=\ell(n-1)-t_\ell(q).
  \]

  Fix \(i\in[n]\). If \(W_i=W_{b_2}\), then by construction one has \(z_j(W_i)=1\) exactly for those
  helper nodes \(\mathcal{H}_j\) indexed by \(C_2\), and \(z_j(W_i)=0\) for all other helper nodes.
  Hence
  \[
    \sum_{j\ne i} z_j(W_i)=t_\ell(q).
  \]
  The case \(W_i=W_{b_1}\) is identical. Therefore \(\lambda_i\ge t_\ell(q)\) for every \(i\in[n]\),
  and so
  \[
    \gamma_i(\mathcal{C})\le \ell(n-1)-t_\ell(q).
  \]
  On the other hand, always \(\gamma_i(\mathcal{C})\ge \beta_i(\mathcal{C})\), and we already proved
  that \(\beta_i(\mathcal{C})=\ell(n-1)-t_\ell(q)\). Thus
  \[
    \gamma_i(\mathcal{C})=\ell(n-1)-t_\ell(q)\qquad\text{for all }i\in[n].
  \]
  Consequently,
  \[
    \gamma_{\avg}(\mathcal{C})=\gamma_{\max}(\mathcal{C})=\ell(n-1)-t_\ell(q).
  \]
\end{proof}

\subsection{On the Necessity of the Length Condition}

Computational evidence suggests that, even below the length range in
Theorem~\ref{thm:two-parity-attainability-general}, the projective counting bound can still be
attained. Nevertheless, in the case \(r=\ell=2\), if one further assumes that all node lines lie in a
fixed regular spread, then the same length condition becomes necessary as well.

Set \[ J:=2(n-1)-\frac{q^2-1}{q-1}=2(n-1)-(q+1), \] which is the projective counting lower bound in the case \(r=\ell=2\). 

\begin{proposition}\label{prop:regular-spread-necessary-length} Assume that \(q\ge 3\) and \(r=\ell=2\). Let \(\mathcal S\) be a regular spread of \(\PG(3,q)\), and let \(\mathcal C\) be an \((n,n-2,2)\) MDS array code over \(\F_q\) whose induced node lines all lie in \(\mathcal S\). If \[ \beta_{\avg}(\mathcal C)=J \qquad\text{or}\qquad \beta_{\max}(\mathcal C)=J, \] then \[ n\ge \min\{2q+2,\,3q-3\}. \] \end{proposition} 

\begin{proof} 
  Let \(\mathcal H_1,\dots,\mathcal H_n\le \mathbb V=\F_q^4\) be the induced node subspaces, and write \[ L_i:=\mathbb P(\mathcal H_i)\in \mathcal S \qquad (i\in[n]). \] Since \(\mathcal C\) is MDS and \(r=2\), the subspaces \(\mathcal H_1,\dots,\mathcal H_n\) are pairwise disjoint, hence the lines \(L_1,\dots,L_n\) are pairwise skew. 
  
  By the projective counting bound specialized to \(r=\ell=2\), one has $\beta_i(\mathcal C)\ge J$ for every $i\in[n]$. Therefore, if either \(\beta_{\avg}(\mathcal C)=J\) or \(\beta_{\max}(\mathcal C)=J\), then in fact $\beta_i(\mathcal C)=J$ for all $i\in[n]$. Equivalently, $\alpha_i=q+1$ for all $i\in[n]$. Hence for each \(i\in[n]\) there exists \(W_i\in\mathcal W_i\) such that \[ \sum_{j\ne i}\dim(W_i\cap \mathcal H_j)=q+1. \] Since this attains equality in Lemma~\ref{lem:projective-counting-bound}, the equality discussion in Remark~\ref{rem:necessary-length-for-equality} shows that \[ \dim(W_i\cap \mathcal H_j)\in\{0,1\} \qquad\text{for all }j\ne i. \] 
  
  Define the hit set \[ B_i:=\{\,j\in[n]\setminus\{i\}: \dim(W_i\cap \mathcal H_j)=1\,\}. \] Then \(|B_i|=q+1\) and \(i\notin B_i\). Let \(m_i:=\mathbb P(W_i)\). Since \(|B_i|=q+1>0\), the line \(m_i\) meets some line \(L_j\in\mathcal S\). We claim that
  \(m_i\notin\mathcal S\). Indeed, if \(m_i\in\mathcal S\), then as the lines of a spread are pairwise
  skew, the only line of \(\mathcal S\) meeting \(m_i\) would be \(m_i\) itself. Hence \(B_i\) would have
  size at most \(1\), contradicting \(|B_i|=q+1\). Therefore \(m_i\notin\mathcal S\). By Proposition~\ref{prop:transversal-regulus}, the set \[ R(m_i):=\{\,L\in\mathcal S:L\cap m_i\neq\emptyset\,\} \] is a regulus contained in \(\mathcal S\), of size \(q+1\). Since each \(j\in B_i\) satisfies \(m_i\cap L_j\neq\emptyset\), we have \[ \{L_j:j\in B_i\}\subseteq R(m_i). \] Both sides have size \(q+1\), so in fact \[ R(m_i)=\{L_j:j\in B_i\}. \] Let \(\mathcal B:=\{B_i:i\in[n]\}\) be the family of distinct hit sets. If \(B,B'\in\mathcal B\) are distinct, choose \(i,j\in[n]\) with \(B_i=B\) and \(B_j=B'\). Then \(R(m_i)\neq R(m_j)\), and Corollary~\ref{cor:reguli-intersection} gives \[ |B\cap B'|=|R(m_i)\cap R(m_j)|\le 2. \] Finally, for every \(x\in[n]\) one has \(x\notin B_x\), so some member of \(\mathcal B\) omits \(x\). Hence \[ \bigcap_{B\in\mathcal B} B=\emptyset. \] Applying Lemma~\ref{lem:block-intersection-principle} to \(\mathcal B\), we obtain \[ n\ge \min\{2(q+1),\,3(q+1)-6\}=\min\{2q+2,\,3q-3\}. \] \end{proof}

\begin{proof}[Proof of Theorem~\ref{thm:two-parity-attainability-regular-spread}] 
  We first prove \((1)\Rightarrow(2)\). By Theorem~\ref{thm:two-parity-attainability-general}, whose remaining exceptional cases are settled in
  Appendix~\ref{app:exceptional-cases}, for every \(n\) satisfying \[ \min\{2q+2,\,3q-3\}\le n\le q^2+1 \] there exists an \((n,n-2,2)\) MDS array code over \(\F_q\) for which \[ \beta_{\avg}(\mathcal{C})=\beta_{\max}(\mathcal{C}) =\gamma_{\avg}(\mathcal{C})=\gamma_{\max}(\mathcal{C})=J. \] These constructions are obtained by selecting node lines from a standard Desarguesian spread of \(\PG(3,q)\). Since by Proposition~\ref{prop:regular-spread-equivalent}, every regular spread in \(\PG(3,q)\) is projectively equivalent to a standard Desarguesian spread, a projective automorphism of \(\PG(3,q)\) carries the node lines of such a construction into the fixed spread \(\mathcal{S}\). Transporting the intrinsic data by this automorphism preserves the MDS property and all repair parameters. Thus there exists an \((n,n-2,2)\) MDS array code whose induced node lines all lie in \(\mathcal{S}\), and for which in fact all four quantities \[ \beta_{\avg}(\mathcal{C}),\ \beta_{\max}(\mathcal{C}),\ \gamma_{\avg}(\mathcal{C}),\ \gamma_{\max}(\mathcal{C}) \] are equal to \(J\).
  
  We then prove \((2)\Rightarrow(1)\). Suppose that there exists an \((n,n-2,2)\) MDS array code \(\mathcal{C}\) over \(\F_q\) whose induced node lines all lie in \(\mathcal{S}\), and such that at least one of \[ \beta_{\avg}(\mathcal{C}),\ \beta_{\max}(\mathcal{C}),\ \gamma_{\avg}(\mathcal{C}),\ \gamma_{\max}(\mathcal{C}) \] is equal to \(J\). Since the induced node lines satisfy the \(r\)-wise spanning condition with \(r=\ell=2\),
  Lemma~\ref{lem:length-bound} gives
  \[n\le q^2+1.\]
  
  For the lower bound, if either \(\gamma_{\avg}(\mathcal{C})=J\) or \(\gamma_{\max}(\mathcal{C})=J\), then since
\[
  \gamma_{\avg}(\mathcal{C})\ge \beta_{\avg}(\mathcal{C}),
  \qquad
  \gamma_{\max}(\mathcal{C})\ge \beta_{\max}(\mathcal{C}),
\]
and the projective counting bound gives
\[
  \beta_{\avg}(\mathcal{C})\ge J,
  \qquad
  \beta_{\max}(\mathcal{C})\ge J,
\]
the corresponding repair-bandwidth parameter also equals \(J\). Thus, in all cases, either
\(\beta_{\avg}(\mathcal{C})=J\) or \(\beta_{\max}(\mathcal{C})=J\). Proposition~\ref{prop:regular-spread-necessary-length}
therefore gives
\[
  n\ge \min\{2q+2,\,3q-3\}.
\]
Combining this with $n\le q^2+1$, we obtain
\[
  \min\{2q+2,\,3q-3\}\le n\le q^2+1,
\]
which is exactly \((1)\).
\end{proof}
\begin{remark}
  Theorem~\ref{thm:two-parity-attainability-regular-spread} suggests a limitation of the Desarguesian-spread
  constructions used above. It would therefore be interesting to look for different constructions that
  attain the projective counting bound over a wider range of lengths.
\end{remark}

\appendix
\section*{Appendix}
\phantomsection
\addcontentsline{toc}{section}{Appendix}
\setcounter{subsection}{0}
\renewcommand{\thesubsection}{\Alph{subsection}}
\makeatletter
\gdef\theHsection{appendix}%
\makeatother
\numberwithin{theorem}{subsection}

\subsection{Background on Spreads and Reguli} \label{app:finite-geometry} 
We briefly recall the finite-geometric notions used in the main text. We write \(\PG(m,q)\) for the \(m\)-dimensional projective space over \(\F_q\). 

\begin{definition} A \emph{spread} of \(\PG(2\ell-1,q)\) is a family \(\mathcal S\) of \((\ell-1)\)-dimensional projective subspaces that partition the points of \(\PG(2\ell-1,q)\). 
\end{definition} 
\begin{remark} Equivalently, a spread of \(\PG(2\ell-1,q)\) may be viewed as a family of \(\ell\)-dimensional \(\F_q\)-subspaces of \(\F_q^{2\ell}\) whose nonzero vectors partition \(\F_q^{2\ell}\setminus\{0\}\). In particular, if \(\mathcal S\) is a spread, then for any distinct \(H,H'\in\mathcal S\) one has \(H\oplus H'=\F_q^{2\ell}\), and \(|\mathcal S|=q^\ell+1\). 
\end{remark} 

We now recall the standard Desarguesian spread construction. Identify \(\F_q^{2\ell}\) with \(\F_{q^\ell}^2\) as \(\F_q\)-vector spaces. 

\begin{construction}[Standard Desarguesian spread]\label{cons:desarguesian-spread} For each \(a\in\F_{q^\ell}\), define \[ \mathcal H_a:=\{(x,ax):x\in\F_{q^\ell}\}\le \F_{q^\ell}^2, \] and define \[ \mathcal H_\infty:=\{(0,y):y\in\F_{q^\ell}\}\le \F_{q^\ell}^2. \] Interpreting these as \(\F_q\)-subspaces of \(\F_q^{2\ell}\), set \[ \mathcal S_{\mathrm{Des}} :=\{\mathcal H_a:a\in\F_{q^\ell}\}\cup\{\mathcal H_\infty\}. \] 
\end{construction} 

\begin{proposition}\label{prop:desarguesian-spread} The family \(\mathcal S_{\mathrm{Des}}\) is a spread of \(\PG(2\ell-1,q)\). 
\end{proposition} 
\begin{proof} Each \(\mathcal H_a\) and \(\mathcal H_\infty\) is a \(1\)-dimensional \(\F_{q^\ell}\)-subspace of \(\F_{q^\ell}^2\), hence an \(\ell\)-dimensional \(\F_q\)-subspace of \(\F_q^{2\ell}\). If \(a,b\in\F_{q^\ell}\) with \(a\neq b\), then \[ (x,ax)=(y,by)\in \mathcal H_a\cap \mathcal H_b \] implies \(x=y\) and \((a-b)x=0\), hence \(x=0\). Thus \(\mathcal H_a\cap \mathcal H_b=\{0\}\). Likewise, \(\mathcal H_a\cap \mathcal H_\infty=\{0\}\) for every \(a\in\F_{q^\ell}\). Finally, let \((u,v)\neq(0,0)\) in \(\F_{q^\ell}^2\). If \(u=0\), then \((u,v)\in \mathcal H_\infty\). If \(u\neq 0\), then \((u,v)\in \mathcal H_{vu^{-1}}\). Hence the nonzero vectors of \(\F_q^{2\ell}\) are partitioned by the members of \(\mathcal S_{\mathrm{Des}}\), so \(\mathcal S_{\mathrm{Des}}\) is a spread. 
\end{proof} 

When \(\ell=2\), the spread elements are projective lines in \(\PG(3,q)\), equivalently, \(2\)-dimensional subspaces of \(\F_q^4\). In this case we also need the standard notions of reguli and regular spreads. 

\begin{definition}
  Let \(\mathcal L\) be a set of lines in \(\PG(3,q)\). A line \(m\) is called
  a \emph{transversal} of \(\mathcal L\) if \(m\) meets every line of
  \(\mathcal L\). A nonempty set \(\mathcal R\) of pairwise skew lines in \(\PG(3,q)\) is
  called a \emph{regulus} if:
  \begin{enumerate}[label=(\arabic*)]
      \item through each point of each line of \(\mathcal R\) there passes a
      transversal of \(\mathcal R\);
      \item through each point of a transversal of \(\mathcal R\) there passes
      a line of \(\mathcal R\).
  \end{enumerate}
  \end{definition}

  \begin{proposition}
    Let \(L_1,L_2,L_3\) be pairwise skew lines in \(\PG(3,q)\). Then there
    exists a unique regulus containing \(L_1,L_2,L_3\), which we denote by
    \(\mathcal R(L_1,L_2,L_3)\).
    \end{proposition}

\begin{proof}
  See \cite[Theorem~2.4.3]{beutelspacher1998projective}.
\end{proof}

\begin{corollary}\label{cor:reguli-intersection} Two distinct reguli in \(\PG(3,q)\) have at most two lines in common. 
\end{corollary}

\begin{definition}
  A spread \(\mathcal S\) of \(\PG(3,q)\) is called \emph{regular} if for every three distinct lines \(L_1,L_2,L_3\in\mathcal S\), one has \[ \mathcal R(L_1,L_2,L_3)\subseteq \mathcal S. \] \end{definition} 

  \begin{proposition}\label{prop:regular-spread-equivalent}
    A spread of \(\PG(3,q)\) is regular if and only if it is projectively
    equivalent to the spread \(\mathcal S_{\mathrm{Des}}\) with \(\ell=2\).
    \end{proposition}
    \begin{proof}
      See \cite[Theorem~4.128]{HT16}
    \end{proof}

\begin{corollary}\label{cor:desarguesian-regular} When \(\ell=2\), the spread \(\mathcal S_{\mathrm{Des}}\) is regular. \end{corollary}

The following proposition will be used in the proof of Proposition~\ref{prop:regular-spread-necessary-length}.

 \begin{proposition}\label{prop:transversal-regulus}
  Let \(\mathcal S\) be a regular spread of \(\PG(3,q)\), and let \(m\) be a line not in \(\mathcal S\). Then
  \[
    R(m):=\{L\in\mathcal S:L\cap m\neq\emptyset\}
  \]
  is a regulus contained in \(\mathcal S\). In particular, \(|R(m)|=q+1\).
\end{proposition}

\begin{proof}
  Since \(\mathcal S\) is a spread and \(m\notin\mathcal S\), each point
  \(P\in m\) lies on a unique line \(L_P\in\mathcal S\), and distinct
  points of \(m\) determine distinct lines of \(\mathcal S\). Hence
  \[
    R(m)=\{L_P:P\in m\},
  \]
  and in particular \(|R(m)|=q+1\).

  Choose three distinct points \(P_1,P_2,P_3\in m\), and set
  \(L_i:=L_{P_i}\in R(m)\) for \(i=1,2,3\). Since \(\mathcal S\) is a spread,
  the lines \(L_1,L_2,L_3\) are pairwise skew. As \(\mathcal S\) is regular, the unique regulus
  \[
    \mathcal R:=\mathcal R(L_1,L_2,L_3)
  \]
  is contained in \(\mathcal S\).

  By the definition of a regulus, there is a transversal \(t\) of
  \(\mathcal R\) through \(P_1\). Then \(t\) meets \(L_2\) and \(L_3\). As
  \(m\) also passes through \(P_1\) and meets \(L_2\) and \(L_3\), and through
  the fixed point \(P_1\) there is at most one line meeting both skew lines
  \(L_2\) and \(L_3\), we obtain \(t=m\). Thus \(m\) is a transversal of
  \(\mathcal R\).

  Now every line of \(\mathcal R\) meets \(m\), so \(\mathcal R\subseteq R(m)\).
  Conversely, let \(P\in m\). Since \(m\) is a transversal of \(\mathcal R\),
  there passes a line of \(\mathcal R\) through \(P\). As \(\mathcal R\subseteq
  \mathcal S\) and \(\mathcal S\) is a spread, this line must be the unique
  spread line \(L_P\) through \(P\). Hence \(L_P\in\mathcal R\), and so
  \(R(m)\subseteq\mathcal R\). Therefore
  \[
    R(m)=\mathcal R,
  \]
  and \(R(m)\) is a regulus contained in \(\mathcal S\).
\end{proof}

\subsection{The Remaining Exceptional Cases in Theorem~\ref{thm:two-parity-attainability-general}} \label{app:exceptional-cases} 
It remains to prove Theorem~\ref{thm:two-parity-attainability-general} in the three cases \[ (q,\ell,n)\in\{(3,2,6),(3,2,7),(4,2,9)\}. \] Throughout this appendix we set \(r=\ell=2\), so \[ t:=t_2(q)=\frac{q^2-1}{q-1}=q+1, \] and identify \(\mathbb{V}=\F_q^4\) with \(\F_{q^2}^2\) as \(\F_q\)-vector spaces.

It is convenient to work with the \emph{conjugate} standard Desarguesian spread \[ \widetilde{\mathcal{H}}_s:=\{(sx,x):x\in\F_{q^2}\}\le \mathbb{V} \qquad (s\in\F_{q^2}), \qquad \widetilde{\mathcal{H}}_\infty:=\{(x,0):x\in\F_{q^2}\}\le \mathbb{V}, \] indexed by \(\mathbb{P}^1(\F_{q^2})=\F_{q^2}\cup\{\infty\}\). For a \(2\)-dimensional \(\F_q\)-subspace \(W\le \mathbb{V}\), define its hit set by
\[
  B(W):=\{\,s\in\mathbb{P}^1(\F_{q^2}): W\cap \widetilde{\mathcal{H}}_s\neq\{0\}\,\}.
\]
Since both \(W\) and \(\widetilde{\mathcal{H}}_s\) are \(2\)-dimensional, one has
\(\dim(W\cap \widetilde{\mathcal{H}}_s)\in\{0,1,2\}\), and the value \(2\) occurs precisely when
\(W=\widetilde{\mathcal{H}}_s\). In particular, if \(W\) is not a spread element, then every nonzero
intersection \(W\cap \widetilde{\mathcal{H}}_s\) is \(1\)-dimensional. 

Hence, if \(W\) is not a spread element, then for every subset
\(\Omega\subseteq \mathbb{P}^1(\F_{q^2})\), every \(s_0\in \Omega\) with
\(W\cap \widetilde{\mathcal{H}}_{s_0}=\{0\}\), and every choice of projective column sets
\(X_s\subseteq \mathbb{P}(\widetilde{\mathcal{H}}_s)\), one has
\[
  \sum_{\substack{s\in\Omega\\ s\ne s_0}} \dim(W\cap \widetilde{\mathcal{H}}_s)
  =
  |B(W)\cap(\Omega\setminus\{s_0\})|,
\]
and
\[
  \sum_{\substack{s\in\Omega\\ s\ne s_0}} z_s(W)
  =
  \sum_{\substack{s\in\Omega\\ s\ne s_0}} |X_s\cap \mathbb{P}(W)|.
\]

Let \(W^{(1)}:=\F_q^2\le \F_{q^2}^2\). Then \(B(W^{(1)})=\mathbb{P}^1(\F_q)\). More generally, for \(g\in \GL_2(\F_{q^2})\), write \(W^{(g)}:=g(W^{(1)})\). Since \(g\in\GL_2(\F_{q^2})\) sends each spread element \(\widetilde{\mathcal H}_s\) to
\(\widetilde{\mathcal H}_{g\cdot s}\), where \(g\cdot s\) is the usual fractional linear action on
\(\mathbb P^1(\F_{q^2})\), one has
\[
  B(W^{(g)})=g\cdot B(W^{(1)})=g\cdot \mathbb P^1(\F_q).
\]

\begin{proposition}\label{prop:exceptional-q3} For \(q=3\) and \(n\in\{6,7\}\), there exists an \((n,n-2,2)\) MDS array code \(\mathcal{C}\) over \(\F_3\) such that \[ \beta_{\avg}(\mathcal{C})=\beta_{\max}(\mathcal{C}) =\gamma_{\avg}(\mathcal{C})=\gamma_{\max}(\mathcal{C}) =2(n-1)-4. \] \end{proposition} \begin{proof} Fix \(\F_9=\F_3(\omega)\) with \(\omega^2=-1\), and consider the following three subsets of \(\mathbb P^1(\F_9)\): \[ B_1=\{\infty,0,1,2\},\qquad B_2=\{\infty,1,1+\omega,1-\omega\},\qquad B_3=\{0,2,1+\omega,1-\omega\}. \] Let \(W^{(1)}:=\F_3^2\), and define \[ g_2=\begin{pmatrix}\omega&1\\0&1\end{pmatrix}, \qquad g_3=\begin{pmatrix}0&-\omega\\1&\omega\end{pmatrix}, \qquad W^{(2)}:=W^{(g_2)},\quad W^{(3)}:=W^{(g_3)}. \] A direct computation shows that \(B(W^{(i)})=B_i\) for \(i=1,2,3\).
  
  For $n=6$, let \[ \Omega_6:=B_1\cup B_2\cup B_3 =\{\infty,0,1,2,1+\omega,1-\omega\}. \] Since \(B_1\cap B_2\cap B_3=\emptyset\), for each \(s\in\Omega_6\) we may choose \(r(s)\in\{1,2,3\}\) such that \(s\notin B_{r(s)}\), and set \(W_s:=W^{(r(s))}\). Then \(W_s\cap \widetilde{\mathcal{H}}_s=\{0\}\), while \[ \sum_{\substack{u\in\Omega_6\\ u\ne s}} \dim(W_s\cap \widetilde{\mathcal{H}}_u) = |B_{r(s)}| = 4 = t \] for every \(s\in\Omega_6\). Hence \(\alpha_i\ge t=4\) for all six selected node subspaces. By
  Lemma~\ref{lem:projective-counting-bound}, one also has \(\alpha_i\le t\), and therefore
  \(\alpha_i=t=4\).
  
  Moreover, each point of \(\Omega_6\) lies in at most two of \(B_1,B_2,B_3\). Therefore, for each \(s\in\Omega_6\), the set \[ Y_s:= \{\,\mathbb{P}(W_u\cap \widetilde{\mathcal{H}}_s): u\in\Omega_6,\ u\ne s,\ W_u\cap \widetilde{\mathcal{H}}_s\ne\{0\}\,\} \subseteq \mathbb{P}(\widetilde{\mathcal{H}}_s) \] has size at most \(2=\ell\). Choose \(X_s\subseteq \mathbb{P}(\widetilde{\mathcal{H}}_s)\) to be any set of two distinct points spanning \(\widetilde{\mathcal{H}}_s\) and containing \(Y_s\). By the converse statements in Section~\ref{sec:intrinsic}, the data \(\bigl(\widetilde{\mathcal{H}}_s,X_s\bigr)_{s\in\Omega_6}\) are realized by a \((6,4,2)\) MDS array code \(\mathcal{C}\). For every failed node \(s\in\Omega_6\), the choice \(X_u\supseteq Y_u\) ensures that
  \(\sum_{u\ne s} z_u(W_s)=4=t\). Hence \[ \beta_{\avg}(\mathcal{C})=\beta_{\max}(\mathcal{C}) =\gamma_{\avg}(\mathcal{C})=\gamma_{\max}(\mathcal{C}) =2(6-1)-4. \] 
  
  For \(n=7\), choose any \(s_\ast\in \mathbb{P}^1(\F_9)\setminus \Omega_6\), and set
\(\Omega_7:=\Omega_6\cup\{s_\ast\}\). Since \(|\mathbb{P}^1(\F_9)|=10\) and \(|\Omega_6|=6\), such a choice is possible; for instance, \(s_\ast=\omega\). Keep the same \(W_s\) for \(s\in\Omega_6\), and set \(W_{s_\ast}:=W^{(1)}\). Then \(W_{s_\ast}\cap \widetilde{\mathcal{H}}_{s_\ast}=\{0\}\), and still \[ \sum_{\substack{u\in\Omega_7\\ u\ne s}} \dim(W_s\cap \widetilde{\mathcal{H}}_u) =4=t \] for every \(s\in\Omega_7\). Since \(s_\ast\) lies in none of \(B_1,B_2,B_3\), the same argument as above gives \(|Y_s|\le 2\) for all \(s\in\Omega_7\). Hence the same intrinsic argument yields a \((7,5,2)\) MDS array code \(\mathcal{C}\), and again
\(\sum_{u\ne s} z_u(W_s)=4=t\) for every failed node \(s\in\Omega_7\). By the same argument as in the case \(n=6\), we obtain \(\alpha_i=t=4\) for all selected nodes, and
hence
\[
  \beta_{\avg}(\mathcal{C})=\beta_{\max}(\mathcal{C})
  =\gamma_{\avg}(\mathcal{C})=\gamma_{\max}(\mathcal{C})
  =2(7-1)-4.
\]\end{proof}

  \begin{proposition}\label{prop:exceptional-q4} For \(q=4\) and \(n=9\), there exists a \((9,7,2)\) MDS array code \(\mathcal{C}\) over \(\F_4\) such that \[ \beta_{\avg}(\mathcal{C})=\beta_{\max}(\mathcal{C}) =\gamma_{\avg}(\mathcal{C})=\gamma_{\max}(\mathcal{C}) =2(9-1)-5. \] 
  \end{proposition} 
  \begin{proof} Fix \(\F_{16}=\F_2(\alpha)\) with \(\alpha^4+\alpha+1=0\), and set \(\beta:=\alpha^5\), so \(\F_4=\{0,1,\beta,\beta^2\}\subseteq \F_{16}\). Consider the following three subsets of \(\mathbb{P}^1(\F_{16})\): \[ B_1=\{\infty,0,1,\beta,\beta^2\}, \] \[ B_2=\{\infty,0,\alpha,\alpha\beta,\alpha\beta^2\}, \qquad B_3=\{\alpha,\beta,\beta^2,\alpha^3,\alpha\beta^2\}. \] Let \(W^{(1)}:=\F_4^2\), and define \[ g_2=\begin{pmatrix}\alpha&0\\0&1\end{pmatrix}, \qquad g_3=\begin{pmatrix}1&\alpha\\ \alpha+1&1\end{pmatrix}, \qquad W^{(2)}:=W^{(g_2)},\quad W^{(3)}:=W^{(g_3)}. \] A direct computation shows that \(B(W^{(i)})=B_i\) for \(i=1,2,3\).
    
    Let \[ \Omega_9:=B_1\cup B_2\cup B_3 =\{\infty,0,1,\beta,\beta^2,\alpha,\alpha\beta,\alpha\beta^2,\alpha^3\}. \] Again \(B_1\cap B_2\cap B_3=\emptyset\), so for each \(s\in\Omega_9\), we may choose \(r(s)\in\{1,2,3\}\) with \(s\notin B_{r(s)}\), and set \(W_s:=W^{(r(s))}\). Then \(W_s\cap \widetilde{\mathcal{H}}_s=\{0\}\), while \[ \sum_{\substack{u\in\Omega_9\\ u\ne s}} \dim(W_s\cap \widetilde{\mathcal{H}}_u) = |B_{r(s)}| = 5 = t \] for every \(s\in\Omega_9\). Hence \(\alpha_i\ge t=5\) for all nine selected node subspaces. By
    Lemma~\ref{lem:projective-counting-bound}, one also has \(\alpha_i\le t\), and therefore
    \(\alpha_i=t=5\). Also, each point of \(\Omega_9\) lies in at most two of \(B_1,B_2,B_3\). Therefore, for each \(s\in\Omega_9\), the corresponding set \[ Y_s:= \{\,\mathbb{P}(W_u\cap \widetilde{\mathcal{H}}_s): u\in\Omega_9,\ u\ne s,\ W_u\cap \widetilde{\mathcal{H}}_s\ne\{0\}\,\} \] has size at most \(2=\ell\). 
    
    Choose \(X_s\subseteq \mathbb{P}(\widetilde{\mathcal{H}}_s)\) to be any set of two distinct points spanning \(\widetilde{\mathcal{H}}_s\) and containing \(Y_s\). By the converse statements in Section~\ref{sec:intrinsic}, these data are realized by a
    \((9,7,2)\) MDS array code \(\mathcal{C}\). For every failed node \(s\in\Omega_9\), the choice \(X_u\supseteq Y_u\) ensures that
    \(\sum_{u\ne s} z_u(W_s)=5=t\). The conclusion now follows exactly as in the case \(q=3\).
  \end{proof}

  \subsection{A Combinatorial Lemma}
\label{app:combinatorial-lemma}

We need the following combinatorial lemma to prove Theorem~\ref{thm:two-parity-attainability-regular-spread}.

\begin{lemma}\label{lem:block-intersection-principle} Let \(\mathcal B\) be a family of \(t\)-subsets of \([n]\) such that \[ \bigcap_{B\in\mathcal B} B=\emptyset \] and \[ |B\cap B'|\le 2 \qquad\text{for all distinct }B,B'\in\mathcal B. \] Then \[ n\ge \min\{2t,\,3t-6\}. \] \end{lemma} \begin{proof} Choose a subfamily \(\{B_1,\dots,B_m\}\subseteq\mathcal B\) that is minimal with empty intersection, i.e. \[ B_1\cap\cdots\cap B_m=\emptyset, \qquad \bigcap_{s\ne r} B_s\neq\emptyset\ \ \text{for every }r\in[m]. \] For each \(r\in[m]\), pick \[ x_r\in \Bigl(\bigcap_{s\ne r} B_s\Bigr)\setminus B_r. \] Then \(x_1,\dots,x_m\) are pairwise distinct. Moreover, for any distinct \(i,j\in[m]\), every \(x_r\) with \(r\notin\{i,j\}\) lies in \(B_i\cap B_j\), so \[ m-2\le |B_i\cap B_j|\le 2. \] Hence \(m\le 4\). Since \(\bigcup_{r=1}^m B_r\subseteq [n]\), it suffices to lower-bound \(|\bigcup_{r=1}^m B_r|\). If \(m=2\), then \(B_1\cap B_2=\emptyset\), so \[ \Bigl|\bigcup_{r=1}^2 B_r\Bigr|=2t. \] If \(m=3\), then by inclusion--exclusion and \(|B_i\cap B_j|\le 2\), \[ \Bigl|\bigcup_{r=1}^3 B_r\Bigr| \ge 3t-\sum_{1\le i<j\le 3}|B_i\cap B_j| \ge 3t-6. \] If \(m=4\), then the inequality \(m-2\le |B_i\cap B_j|\le 2\) forces \(|B_i\cap B_j|=2\) for all \(i\ne j\). For \(\{i,j,k,\ell\}=\{1,2,3,4\}\), the points \(x_k,x_\ell\) both lie in \(B_i\cap B_j\), hence \[ B_i\cap B_j=\{x_k,x_\ell\}. \] In particular, any element of \([n]\setminus\{x_1,x_2,x_3,x_4\}\) belongs to at most one of \(B_1,\dots,B_4\). Also, by construction, each \(B_r\) contains exactly three of \(x_1,x_2,x_3,x_4\). Therefore each \(B_r\) contributes at least \(t-3\) elements outside \(\{x_1,x_2,x_3,x_4\}\) that lie in no other \(B_{r'}\). Consequently,
  \[
    \Bigl|\bigcup_{r=1}^4 B_r\Bigr|
    \ge 4+4(t-3)=4t-8.
  \]
  Since each \(B_r\) contains the three distinct points \(\{x_s:s\ne r\}\), one has \(t\ge 3\), and hence
  \[
    4t-8\ge 3t-6.
  \] In all cases, \[ n\ge \Bigl|\bigcup_{r=1}^m B_r\Bigr|\ge \min\{2t,\,3t-6\}. \] \end{proof} 


%

\phantomsection
\addcontentsline{toc}{section}{\refname}
\bibliographystyle{alphaurl}
\bibliography{references}


\end{document}